\documentclass[11pt]{article}
\usepackage{fullpage}
\usepackage{amsmath,amsthm,amssymb,braket,color}
\usepackage{hyperref}
\hypersetup{colorlinks=true, citecolor=blue, linkcolor=red, urlcolor=blue}
\usepackage[capitalise]{cleveref}
\usepackage{tikz}
\usetikzlibrary{quantikz2}
\tikzset{
    mygroup/.style={rounded corners, draw=gray!30, fill=orange!10, inner xsep=1pt},
    mybox1/.style={draw=gray!30, fill=gray!30, inner xsep=1pt},
    mybox2/.style={draw=gray!15, fill=gray!15},
    mylabel1/.style={label position=above, anchor=north, yshift=10pt},
    mylabel2/.style={label position=above, anchor=north, yshift=4pt}
    }

\usepackage{subcaption}

\newtheorem{theorem}{Theorem}

\newtheorem{fact}[theorem]{Fact}
\newtheorem{corollary}[theorem]{Corollary}
\newtheorem{lemma}[theorem]{Lemma}
\newtheorem{conjecture}[theorem]{Conjecture}
\newtheorem{definition}[theorem]{Definition}

\title{Quantum circuit for multi-qubit Toffoli gate with optimal resource}
\author{
    Junhong Nie
    \thanks{State Key Lab of Processors, Institute of Computing Technology, Chinese Academy of Sciences, Beijing 100190, China}
    \thanks{School of Computer Science and Technology, University of Chinese Academy of Sciences, Beijing 100049, China}
    \\ \texttt{niejunhong19z@ict.ac.cn}
    \and
    Wei Zi
    \footnotemark[1]
    \footnotemark[2]
    \\ \texttt{ziwei20z@ict.ac.cn}
    \and
    Xiaoming Sun
    \footnotemark[1]
    \footnotemark[2]
    \\ \texttt{sunxiaoming@ict.ac.cn}
}

\begin{document}

\maketitle

\begin{abstract}
    Resource consumption is an important issue in quantum information processing, particularly during the present NISQ era. In this paper, we investigate resource optimization of implementing multiple controlled operations, which are fundamental 
    building blocks in the field of quantum computing and quantum simulation. We design new quantum circuits for the $n$-Toffoli gate and general multi-controlled unitary, which have only $O(\log n)$-depth and $O(n)$-size, and only require $1$ ancillary qubit. To achieve these results, we 
    explore the potential of ancillary qubits and discover a method to create new conditional clean qubits from existed ancillary qubits. These techniques can also be utilized to construct an efficient quantum circuit for incrementor, leading to an implementation of multi-qubit Toffoli gate with a depth of $O(\log^2n)$ and size of $O(n)$ without any ancillary qubits. Furthermore, we explore the power of ancillary qubits from the perspective of 
    resource theory. We demonstrate that without the assistance of ancillary qubit, any quantum circuit implementation of multi-qubit Toffoli gate must employ exponential precision gates. This finding indicates a significant disparity in computational power of quantum circuits between using and not using ancillary qubits. Additionally, we discuss the comparison of the power of ancillary qubits and extra energy levels in quantum circuit design.
\end{abstract}

\section{Introduction}
Quantum circuit is one of the most important model in quantum computing. Unitary operations are usually implemented by quantum circuits which consists of CNOT and single-qubit gates. Optimization of circuit depth, size and ancilla count plays a crucial role in the era of Noisy Intermediate-Scale Quantum (NISQ) devices \cite{preskill2018quantum}. This paper introduces a novel quantum circuit implementation of $n$-Toffoli gate, which applies $X$ gate to the target qubit exclusively when all other qubits are at state $\ket{1}$. The application scope of $n$-Toffoli gate extends to diverse fields, including quantum error correction \cite{cory1998experimental,dennis2001toward,PhysRevLett.77.793}, quantum machine learning \cite{rebentrost2014quantum, tacchino2019artificial}, quantum singular value transformation \cite{gilyen2019quantum,PRXQuantum.2.040203}, quantum simulation \cite{babbush2018encoding}, and the implementation of quantum algorithms such as state preparation \cite{PhysRevLett.129.230504}, unitary synthesis \cite{sun2023asymptotically,PhysRevLett.92.177902} and Grover's algorithm \cite{grover1996fast}. The $n$-Toffoli gate demonstrates its usefulness and significance across a wide range of domains, showcasing its versatility and potential impact. 

To achieve systematic approach of quantum computing, unitary are often implemented under elementary universal gate sets, such as CNOT and single-qubit gates. The implementation of the $n$-Toffoli gate is widely concerned \cite{PhysRevLett.129.230501,PhysRevLett.86.3907,PhysRevLett.119.160503,PhysRevX.10.021054}. It is worth noting that regardless of the number of ancillary qubits, any implementation of $n$-Toffoli gate must have a depth of $\Omega(\log n)$ and a size of $\Omega(n)$ \cite{fang2006quantum}. By utilizing qutrits, $n$-Toffoli gate can be implemented with depth $O(\log n)$ without the need for ancillary qutrits \cite{ISCAqutrit}. In the case of qubit system, the current best design aiming to eliminate the reliance on ancillary qubits results in linear depth \cite{da2022linear}. On the other hand, the best design that has circuit depth $O(\log n)$ requires $\Omega(n)$ ancillary qubits \cite{he2017decompositions}. Researchers are actively exploring the existence of a quantum circuit implementing $n$-Toffoli gate with asymptotically optimal circuit and space complexity, while minimizing the use of ancillary qubits. This pursuit has generated substantial interest within the field.

In this paper, we consider exact implementation of $n$-Toffoli gate and $n$-controlled $U$ gate using quantum circuit that only consists of single qubit gates and CNOT gate. Our result is that both these two gates can be implemented in $O(\log n)$ depth and $O(n)$ size, with the help of only $1$ ancillary qubit. For $n$-Toffoli and $n$-controlled $U$ gate such that $U^2=I$, the ancillary qubit is enough to be dirty. For $n$-controlled $U$ gate with $U^2\ne I$, the ancillary qubit is required to be clean. To achieve this result, we figure out a novel way of utilizing existed clean ancillary qubits to create new conditionally clean ancillary qubits. What we mean by conditionally clean is that the qubit is at state $\ket{0}$ conditioned on that some existed clean ancillary qubits are at a particular state. This is shown in \cref{sec:construction}.

We believe that this trick may help in designing other quantum circuits too. To raise evidence, we use the technique to design a quantum circuit for incrementor in depth $O(\log^2n)$ and size $O(n)$, using only $1$ clean ancillary qubit. Recent years, designing quantum arithmetic circuits has become a welcomed topic owing to its broad applications in quantum computing, e.g. Shor's algorithm \cite{shor1994algorithms}. A quantum incrementor computes $\ket{x}\to\ket{(x+1)\mod{2^n}}$ for all $x$. To the best of our knowledge, our approach gives the first poly-logarithmic implementation of this operation with $o(n)$ ancillary qubits. As a by-product, this implementation also leads to a quantum circuit for $n$-Toffoli gate with no ancillary qubit that has depth $O(\log^2n)$ and size $O(n)$. Our construction of incrementor circuit is detailed in \cref{subsec:plus_one}, and why it leads to $n$-Toffoli construction is explained in \cref{subsec:toffoli_0_anc}.

Furthermore, we exhibit two relevant discussions. In \cref{sec:precision}, we try initiating some study about the relationship between number of ancillary qubits and gate precision in quantum operation implementations. In his widely known blog \cite{gidney2015}, Gidney shows how to implement $n$-Toffoli gate in size $O(n)$ without any ancillary qubit. While our implementation in \cref{sec:construction} consists of only up to $4$-Toffoli gates, his construction utilizes ``high precision'' gates such as $Z^{1/2^n}$ which is both hard to be realized precisely on NISQ devices and costly to be performed in current fault-tolerant schemes. Gidney further asks that under the constraint of no ancillary qubit, whether $n$-Toffoli gate can be implemented without using these high precision gates. However, it turns out to be impossible. In \cref{sec:precision}, we prove that when restricted to zero ancillary qubits, both $n$-Toffoli gate and incrementor require exponential precision gates. Given the difficulty of achieving high-precision quantum gates in current fault-tolerant schemes \cite{PhysRevLett.110.190502}, our results may provide some guidance for the design of quantum circuits for fault-tolerant quantum computing.

In \cref{sec:qudit}, we discuss the exchange of ancillary qubits and extra energy levels. Extra energy levels help in implementing a variant of multi-qubit operations. For example, Gokhale et al. \cite{ISCAqutrit} utilized qutrits to implement a generalized $n$-Toffoli gate with a depth of $O(\log n)$ without ancillary qutrit. It is also known that incrementor can be implemented on qutrit system in $O(\log^2 n)$ depth without ancillary qutrit \cite{baker2020improved}. Both of these constructions consist of only reversible qutrit gates (compared to the necessity of exponential precision gates proved in \cref{sec:precision}). Yet people fail to find an implementation task that requires deep circuit on qubit system but permits shallow circuit on qutrit system. In \cref{sec:qudit}, we remark that there are efficient exchanges between qubit system and qudit system. Moreover, we define some bounded-space complexity classes from computation theory's perspective and raise some conjecture about these classes.

\section{Preliminaries}\label{sec:preliminaries}
$n$-Toffoli gate is a quantum gate with $n$ control qubits and $1$ target qubit. It computes the function $\textup{AND}_n$ of the $n$ control qubits and add the result to the target qubit. Formally, it is described as
\[\ket{x_1,x_2,\dots,x_n,t}\to\ket{x_1,x_2,\dots,x_n,t\oplus\bigwedge_{i=1}^nx_i}.\]
An illustration for $n=4$ is shown in \cref{subfig:4_Toffoli}.

It can be further generalized to $n$-controlled $U$ gates, that is, a unitary $U$ is performed on the target qubit if and only if all control qubits are at state $\ket{1}$. In this paper, we denote it as $\textup{C}^nU$. An illustration for $n=4$ is shown in \cref{subfig:4_cU}.
\begin{figure}
    \centering
    \begin{subfigure}{0.4\textwidth}
        \centering
        \begin{quantikz}
            \lstick{$x_1$} &\ctrl{1} &\rstick{$x_1$}\\
            \lstick{$x_2$} &\ctrl{1} &\rstick{$x_2$}\\
            \lstick{$x_3$} &\ctrl{1} &\rstick{$x_3$}\\
            \lstick{$x_4$} &\ctrl{1} &\rstick{$x_4$}\\
            \lstick{$t$} &\targ{} &\rstick{$t\oplus\bigwedge_{i=1}^4x_i$}
        \end{quantikz}
        \caption{$4$-Toffoli gate}
        \label{subfig:4_Toffoli}
    \end{subfigure}
    \begin{subfigure}{0.4\textwidth}
        \centering
        \begin{quantikz}
            &\ctrl{1} &\\
            &\ctrl{1} &\\
            &\ctrl{1} &\\
            &\ctrl{1} &\\
            &\gate{U} &
        \end{quantikz}
        \caption{$4$-controlled $U$ gate}
        \label{subfig:4_cU}
    \end{subfigure}
\end{figure}

In quantum computing, people usually decompose unitary into a set of simple and elementary gates. It is generally challenging to physically implement multi-qubit gates on modern quantum computers. To achieve a more systematic approach of quantum computing, unitary are usually decomposed into a set of simple and elementary gates. It is well-known that CNOT and single-qubit gates are universal for quantum computing \cite{nielsen2010quantum}, that is, any unitary can be implemented by a quantum circuit consisting of CNOT and single-qubit gates. In order to clarify mathematical representation, we denote this gate set as $\mathcal{B}_2$. Other interesting universal gate sets include $\mathcal{U}_2$, all two-qubit unitary gates, and $\mathcal{B}_3$, which consists of Toffoli, CNOT and single-qubit gates. The problem we address in this paper is the exact quantum circuit implementation of $n$-Toffoli gates over gate set $\mathcal{B}_2$.

As a quantum resource, ancillary qubits have been extensively utilized in the process of implementing quantum operations, serving to simplify and streamline the overall process.
These qubits can carry intermediate information along the computation process or serve as draft spaces. Since they are expected to reuse later, it is required to recover these ancillary qubits to its initial state after finishing the computation. In principle, any unitary operation $U$ can be implemented without ancillary qubits \cite{sun2023asymptotically}. However, this may results in significant increasing in circuit size and depth. 

In general, there are two kinds of ancillary qubits, clean and dirty. Clean ancillary qubits are ancillary qubits with initial state $\ket{0}$, while dirty ancillary qubits can be at any state initially and should be recovered to its initial state after computation. From the definitions, clean ancillary qubits are harder to obtain and may have stronger ability than dirty ones. However, under some circumstances, use of clean ancillary qubits can be substituted by dirty qubits. \cref{fig:clean_to_dirty} shows an example which will be exploited in this paper for several times. For $U^2=I$, it can be easily verified that the two circuits are equivalent. Notice that in the figure, the implementation of controlled $U$ operation can even use the first register as dirty ancillary qubits.
\begin{figure*}
    \centering
    \begin{quantikz}
        &\qwbundle{} &\ctrl{1} &\\
        & &\gate{U} &
    \end{quantikz}
    $=$
    \begin{quantikz}
        &\qwbundle{} &\ctrl{2} & &\ctrl{2} &\\
        &\qwbundle{} & &\gate{U} & &\\
        \lstick{0} & &\targ{} &\ctrl{-1} &\targ{} &\rstick{0}
    \end{quantikz}
    $\iff$
    \begin{quantikz}
        &\qwbundle{} & &\ctrl{2} & &\ctrl{2} &\\
        &\qwbundle{} &\gate{U} & &\gate{U} & &\\
        \lstick{a} & &\ctrl{-1} &\targ{} &\ctrl{-1} &\targ{} &\rstick{a}
    \end{quantikz}
    \caption{An example for substituting clean ancillary qubits by dirty ones. $U^2=I$.}
    \label{fig:clean_to_dirty}
\end{figure*}
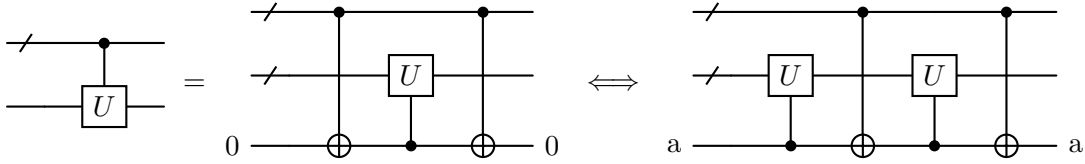

\section{Logarithmic-depth Toffoli with one dirty ancilla}\label{sec:construction}
In this section, we show the promised implementations of $n$-Toffoli gate and $\textup{C}^nU$ gate which have depth $O(\log n)$ with $1$ ancillary qubit. For $n$-Toffoli gate and $\textup{C}^nU$ gate with $U^2=I$, it is enough for the ancilla to be dirty. For $\textup{C}^nU$ gate with $U^2\ne I$, we require it to be clean. Moreover, our constructions keep circuit size $O(n)$. That is, up to one ancillary qubit, our construction is optimal in both size and depth. To approach this, we find a way to create new clean ancillary qubits with the help of existed clean ancillary qubits. Although one cannot expect to turn input qubits into pure clean qubits (which means the qubit is at the particular state $\ket{0}$) due to reversibility, it is somehow possible to create conditionally clean qubits. That is, conditioned that the existed clean qubits are in some partitcular state, some input qubits are clean. This new aspective of utilizing ancillary qubits may be of independent interest in quantum circuit design.

\begin{theorem}\label{thm:Toffoli}
  The $n$-Toffoli gate can be implemented by a quantum circuit over $\mathcal{B}_2$ of depth $O(\log n)$ and size $O(n)$, with the assistance of $1$ clean ancillary qubit. 
\end{theorem}
\begin{proof}
    In order to get rid of rounding, we assume $n$ is even. Proof for odd $n$ is completely the same. The overall construction is shown in \cref{fig:construction}. Step 1, 2 and 3 computes the AND of inputs on target qubit, and step 4 and 5 are merely reverse of step 2 and 1 respectively, thus both inputs and ancillary qubit remain unchanged. $(\frac{n}{2}-2)$-Toffoli gates inside the colored boxes are implemented recursively using the same construction, taking $I_2,I_4$ as clean ancillary qubits and $I_1,I_3$ as target qubits.
    
    For the correctness of our construction, notice that if there is some input among $I_1,I_2,I_3,I_4$ equal to $0$, then after step 1, $A$ will keep at $0$, thus, step 3 will not change the state of $T$.
    If $x_1,x_2,x_3,x_4$ are all equal to $1$, then after step 1, the clean ancillary qubit $A$ will be flipped to $1$, and after the $X$ gates in step 2, $I_1,I_2,I_3,I_4$ are flipped to $0$ which are conditionally clean. After the recursive step, $I_1$ and $I_3$ will store $\textup{AND}(x_5,\dots,x_{\frac{n}{2}+2})$ and $\textup{AND}(x_{\frac{n}{2}+3},\dots,x_n)$ respectively. Finally after step 3, $\textup{AND}(x_5,\dots,x_n)$ will be computed on the target qubit.
    
    Below we analyse the resource assumption of our design. The key point here is that we delay the recovering of input qubits into their original states to the second half. To be concrete, we recursively implement only the first (second) half of the overall construction for these $(\frac{n}{2}-2)$-Toffoli in step 2 (step 4). Since the first half contains step 3, this strategy ensures that we indeed compute the result of these two Toffoli on $I_1$ and $I_3$ while not boosting circuit depth. Now denote the depth of this circuit by $D(n)$. According to \cref{fig:construction}, step 1, 3 and 5 has $O(1)$ depth. And as described above, Toffoli gates in step 2 and 4 only consists of the first and second half of the overall construction respectively, so they share a depth of $D(\frac{n}{2})$. In all, we have
    \[D(n)=D(\frac{n}{2})+O(1),\]
    which results in $D(n)=O(\log n)$. A similar analysis shows that the circuit size is $O(n)$. This concludes the proof.
\end{proof}
\begin{figure*}
    \centering
    \begin{quantikz}
        \lstick{$I_1=x_1$} &\ctrl{1}\slice{1} &\gate{X} &\targ{}\gategroup[8, steps=1, style={mygroup}, background, label style={mylabel1}]{} &\slice{2} &\ctrl{2}\slice{3} & &\targ{}\gategroup[8, steps=1, style={mygroup}, background, label style={mylabel1}]{} &\gate{X}\slice{4} &\ctrl{1}\slice{5} &\\
        \lstick{$I_2=x_2$} &\ctrl{1} &\gate{X} &\push{\diamondsuit}\wire[u][1]{a} & & & &\push{\diamondsuit}\wire[u][1]{a} &\gate{X} &\ctrl{1} &\\
        \lstick{$I_3=x_3$} &\ctrl{1} &\gate{X} & &\targ{}\gategroup[10, steps=1, style={mygroup}, background, label style={mylabel1}]{} &\ctrl{10} &\targ{}\gategroup[10, steps=1, style={mygroup}, background, label style={mylabel1}]{} & &\gate{X} &\ctrl{1} &\\
        \lstick{$I_4=x_4$} &\ctrl{9} &\gate{X} & &\push{\diamondsuit}\wire[u][1]{a} & &\push{\diamondsuit}\wire[u][1]{a} & &\gate{X} &\ctrl{9} &\\
        \lstick{$I_5=x_5$} & & &\ctrl{-3} & & & &\ctrl{-3} & & &\\
        \setwiretype{n}\lstick{\vdots} & & &\vdots & & & &\vdots & & &\\
        \lstick{$I_{\frac{n}{2}+1}=x_{\frac{n}{2}+1}$} & & &\control{} & & & &\control{} & & &\\
        \lstick{$I_{\frac{n}{2}+2}=x_{\frac{n}{2}+2}$} & & &\ctrl{-1} & & & &\ctrl{-1} & & &\\
        \lstick{$I_{\frac{n}{2}+3}=x_{\frac{n}{2}+3}$} & & & &\ctrl{-5} & &\ctrl{-5} & & & &\\
        \setwiretype{n}\lstick{\vdots} & & & &\vdots & &\vdots & & & &\\
        \lstick{$I_{n-1}=x_{n-1}$} & & & &\control{} & &\control{} & & & &\\
        \lstick{$I_n=x_n$} & & & &\ctrl{-1} & &\ctrl{-1} & & & &\\
        \lstick{$A=0$} &\targ{} & & & &\ctrl{1} & & & &\targ{} &\\
        \lstick{$T=t$} & & & & &\targ{} & & & & &
    \end{quantikz}

    \caption{Implementation of $n$-Toffoli gate using one clean ancillary qubit. $\diamondsuit$ represents for ancillary qubit. Step number marks the circuit block to the left of its corresponding barrier.}
    \label{fig:construction}
\end{figure*}
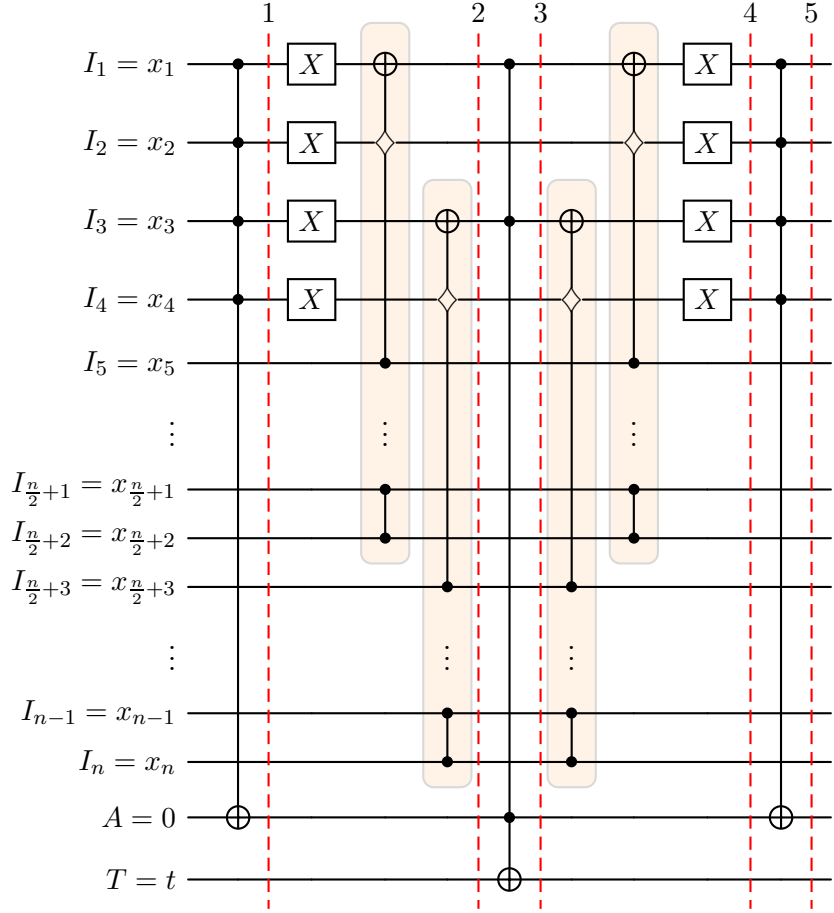

\begin{corollary}
    For any single qubit unitary $U$, $\textup{C}^nU$ gate can be implemented by a quantum circuit over $\mathcal{B}_2$ of depth $O(\log n)$ and size $O(n)$, with the help of $1$ clean ancillary qubit.
\end{corollary}
\begin{proof}
    Simply substitute the $4$-Toffoli gate in step 3 of \cref{fig:construction} by a $\textup{C}^4U$ gate with the same control and target qubits.
\end{proof}

Now we can use the simple trick in \cref{fig:clean_to_dirty} to turn the clean ancillary qubit into dirty one. 
\begin{corollary}
    For any single qubit unitary $U$ such that $U^2=I$, $\textup{C}^nU$ gate can be implemented by a quantum circuit over $\mathcal{B}_2$ of depth $O(\log n)$ and size $O(n)$, with the help of $1$ dirty ancillary qubit.
\end{corollary}
\begin{proof}
    We use the trick shown in the second equivalence of \cref{fig:clean_to_dirty}. In \cref{fig:construction}, mark $I_1,I_2,I_3,I_4$ as the first register, and $A$ as the clean ancillary qubit. The unitary $U$ is simply what is controlled by $A$ in step 2, 3 and 4. Notice that in \cref{fig:clean_to_dirty}, the implementation of $U$ can even use the first register as long as it keeps the state of first register unchanged. The overall sequence is step 2, 3, 4, 1, 2, 3, 4, 5.
\end{proof}

\section{Toffoli without ancilla}\label{sec:toffoli_0_anc}
In this section, we show how to implement $n$-Toffoli gate using an $O(\log ^2n)$ depth quantum circuit without ancillary qubit. We adopt the framework from Gidney \cite{gidney2015} in which quantum incrementor is utilized in an elegant way. Our main contribution is a poly-logarithmic depth circuit for quantum incrementor using $1$ dirty ancillary qubit. The main idea is that we can compute the carry of lower bits using multi-qubit Toffoli gate from \cref{sec:construction}, which meanwhile creates $O(n)$ conditional clean ancilllary qubits, and this allows us to call the conventional logarithmic depth carry-lookahead quantum adder from \cite{draper2004logarithmic}.

\subsection{Quantum incrementor circuit}\label{subsec:plus_one}
Here we describe our poly-logarithmic depth design for incrementor. The circuit computes the following map:
\[\ket{x}\to\ket{(x+1)\mod{2^n}},\]
where $x$ is an $n$-bit binary number. 
\begin{theorem}\label{thm:plus_one}
    Quantum incrementor can be implemented by a quantum circuit over $\mathcal{B}_2$ of depth $O(\log ^2n)$ and size $O(n)$ using $1$ clean ancillary qubit.
\end{theorem}
\begin{proof}
    The overall construction is shown in \cref{fig:plus_one}. In the figure, we use $+1$ gate to represent for incrementor on the corresponding qubits. In the circuit, step 1 and 3 are $\frac{9}{10}n$-Toffoli gate which has depth $O(\log n)$ and size $O(\frac{9}{10}n)$ according to \cref{sec:construction}. Step 2 contains an ancilla controlled incrementor using the lower $\frac{9}{10}n$ qubits as clean ancillary qubits. Since we have enough clean ancillary qubits, we can safely call the carry-lookahead modulo-$2^{\frac{1}{10}n}$ adder from Section 5.1 of \cite{draper2006logarithmic} which has depth $O(\log n)$ and size $O(\frac{1}{10}n)$. But we need the adder to be controlled by $A$ without boosting circuit depth and size. If we had $\frac{1}{10}n$ clean ancillary qubits, we could copy $A$ on these qubits by a fan-out gate (many CNOT gates that share the same control qubit), and, since there are at most $\frac{1}{10}n$ gates for each layer of the adder, we could use each of these copied $A$ to control one gate of the layer. In other words, we could parallelize the first equivalence of \cref{fig:clean_to_dirty}. However, here we only have dirty ancillary qubits. Fortunately, the second equivalence of \cref{fig:clean_to_dirty} comes to rescue. Since we spare up to $\frac{9}{10}n$ space, we still have enough qubits to serve as dirty ancillary qubits. Also, by investigating the detail design of Section 5.1 in \cite{draper2006logarithmic}, not all layers need to be controlled. In fact, the adder is made of several stages: P stage and G stage of $O(\log n)$ depth, and writing stage of $O(1)$ depth. The first two stages only affect ancillary qubits and never change input qubits, so they don't need to be controlled. Only the writing stage need to be controlled by $A$ and the controlled version can be implemented in $O(\log n)$ depth by utilizing \cref{fig:clean_to_dirty} in parallel. In all, we can indeed implement this controlled adder in depth $O(\log n)$ and size $O(\frac{1}{10}n)$. Step 4 is a $\frac{9}{10}n$ size incrementor that is implemented recursively using $A$ as clean ancillary qubit. Let $D(n)$ be circuit depth and $S(n)$ be circuit size. According to analysis above, we have
    \begin{align*}
        D(n)&=D(\frac{9}{10}n)+O(\log n);\\
        S(n)&=S(\frac{9}{10}n)+O(\frac{9}{10}n)+O(\frac{1}{10}n).
    \end{align*}
    Solving these recursions, we have $D(n)=O(\log^2 n)$ and $S(n)=O(n)$.

    We prove the correctness by looking at the value of $A$ at step 2. If $A=0$ at step 2, then there are $0$'s among lower bits of the input, and the incrementor on $n$ bits is equivalent to a smaller incrementor on the lower part which has no carry. In the circuit, step 1 and 3 cancel, and step 4 indeed computes a smaller incrementor.
    If $A=1$ at step 2, then all the lower bits are equal to $1$, which means incrementor should produce a carry from the lower part to the higher part. In fact, it is equivalent to a smaller incrementor on the higher part which is indeed computed at step 2 in our circuit. Up to now, the higher parts are settled. For the lower parts, as stated above, step 4 computes incrementor on the lower parts and the result is all zero as expected. This concludes the proof.
\end{proof}
\begin{figure}
    \centering
    \begin{quantikz}        \lstick{$x_{[1\dots\frac{9}{10}n]}$} &\qwbundle{\frac{9}{10}n} &\ctrl{2}\slice{1} &\gate{X} &\push{\diamondsuit}\wire[d][1]{a} &\gate{X}\slice{2} &\ctrl{2}\slice{3} &\gate{+1}\slice{4} &\\
        \lstick{$x_{[\frac{9}{10}n\dots n]}$} &\qwbundle{\frac{1}{10}n} & & &\gate{+1} & & & &\\
        \lstick{$A=0$} & &\targ{} & &\ctrl{-1} & &\targ{} &\push{\diamondsuit}\wire[u][2]{a} &
    \end{quantikz}
    \caption{Quantum circuit for incrementor. $\diamondsuit$ represents for ancillary qubit. The final incrementor is implemented recursively using $A$ as clean ancillary qubit. Step number marks the circuit block to the left of its corresponding barrier.}
    \label{fig:plus_one}
\end{figure}
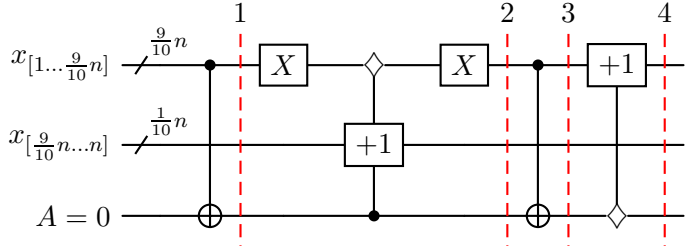

\subsection{Constructing Toffoli}\label{subsec:toffoli_0_anc}
Here we give the final construction for multi-qubit Toffoli gate. Briefly, we copy Gidney's beautiful construction \cite{gidney2015} and replace the $+1$ gate by our design in \cref{thm:plus_one}. However, we do not have one clean ancillary qubit required in \cref{thm:plus_one}. To settle this, we create a conditionally clean qubit from scratch. Before presenting our construction, we show that fan-out like $Z(\theta)$ gates can be implemented in parallel, where
\[Z(\theta)=\begin{bmatrix}
    1 &\\
     &\exp(i\theta)
\end{bmatrix}\]
is the conditional phase operator.
\begin{lemma}\label{lem:rz_fanout}
    $n$ arbitrary $Z(\theta_i)$ gates acting on $n$ qubits which are controlled by the same qubit can be implemented by a quantum circuit over $\mathcal{B}_2$ of depth $O(\log n)$ and size $O(n)$ using no ancillary qubit.
\end{lemma}
\begin{proof}
    The overall construction is shown in \cref{fig:rz_fanout}. In the figure, we use $\gamma$ as an abbreviation to $Z(\gamma)$. The parameter $\theta$ is set to be equal to $\sum_{i=1}^n\frac{\theta_i}{2}$. It is well-known that the two fan-out gates can be implemented in depth $O(\log n)$ and size $O(n)$, so the resource assumption is as claimed. For the $i$-th qubit with input $x_i$, it is unchanged if $c=0$, and it accumulate a phase of
    \[x_i\cdot\frac{\theta_i}{2}+(1-x_i)\cdot\left(-\frac{\theta_i}{2}\right)=x_i\cdot\theta_i-\frac{\theta_i}{2}\]
    if $c=1$. The first phase is as desired, and the second rotation is compensated by a phase of $\theta$ if $c=1$.
\end{proof}
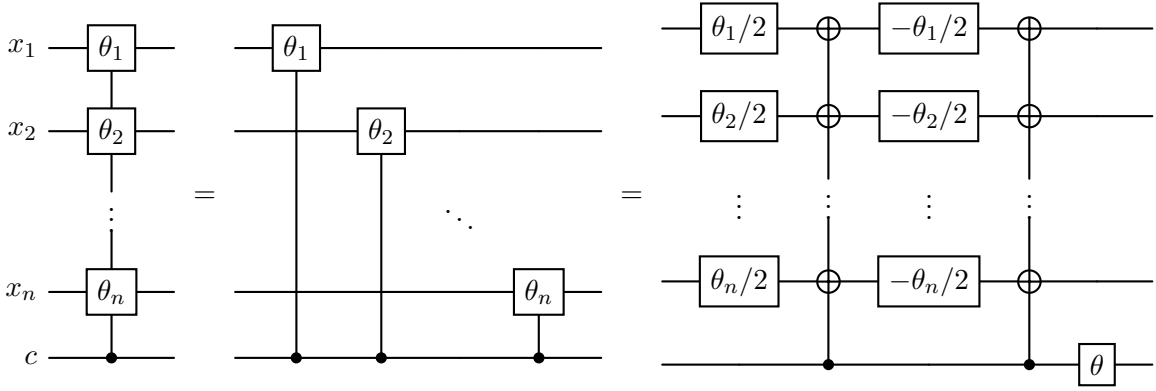
\begin{figure*}[t]
    \centering
    
    \begin{quantikz}
        \lstick{$x_1$} &\gate{\theta_1}\wire[d][1]{a} &\\
        \lstick{$x_2$} &\gate{\theta_2}\wire[d][1]{a} &\\
        \setwiretype{n}&\vdots &\\
        \lstick{$x_n$} &\gate{\theta_n}\wire[u][1]{a} &\\
        \lstick{$c$} &\ctrl{-1} &
    \end{quantikz}
    $=$
    \begin{quantikz}
        &\gate{\theta_1} & & & &\\
        & &\gate{\theta_2} & & &\\
        \setwiretype{n}& & &\ddots & &\\
        & & & &\gate{\theta_n} &\\
        &\ctrl{-4} &\ctrl{-3} & &\ctrl{-1} &
    \end{quantikz}
    $=$
    \begin{quantikz}
        &\gate{\theta_1/2} &\targ{}\wire[d][1]{a} &\gate{-\theta_1/2} &\targ{}\wire[d][1]{a} & &\\
        &\gate{\theta_2/2} &\targ{}\wire[d][1]{a} &\gate{-\theta_2/2} &\targ{}\wire[d][1]{a} & &\\
        \setwiretype{n}&\vdots &\vdots &\vdots &\vdots & &\\
        &\gate{\theta_n/2} &\targ{}\wire[u][1]{a} &\gate{-\theta_n/2} &\targ{}\wire[u][1]{a} & &\\
        & &\ctrl{-1} & &\ctrl{-1} &\gate{\theta} &
    \end{quantikz}
    \caption{Implementation of $n$ arbitrary $R_z$ gates controlled by the same qubit.}
    \label{fig:rz_fanout}
\end{figure*}

Now we are ready to present the final construction.
\begin{theorem}
    $n$-Toffoli gate can be implemented by a quantum circuit over $\mathcal{B}_2$ of depth $O(\log^2n)$ and size $O(n)$ using no ancillary qubit.
\end{theorem}
\begin{proof}
    The overall construction is shown in \cref{fig:gidney}. Circuit inside the colored box without control is the original construction from Gidney for $(n-1)$-Toffoli gate \cite{gidney2015}. The $n$-th input bit $x_n$ controls all gates inside colored box except for the $+1$ and $-1$ gates. In particular, all the $R_z$ gates are controlled by $x_n$. The $+1$ gate adopts our design in \cref{thm:plus_one} using $x_n$ as clean ancillary qubit and $-1$ gate is merely the reverse of $+1$. For correctness of this circuit, consider different values of $x_n$. If $x_n=0$, no operations inside colored box are performed except for $+1$ and $-1$, while $+1$ and $-1$ cancel. If $x_n=1$, all operations inside colored box are performed. $+1$ and $-1$ gates function correctly, because the $n$-th bit is ensured to be of value $1$ when they are performed and it is conditionally clean. (Note that \cref{thm:plus_one} requires a clean ancillary qubit that has value $0$, and this is easily achieved by adding two $X$ gates. We omit this detail in \cref{fig:gidney} due to limitation of space.) Gidney ensures us the correctness in this case.

    This circuit has depth $O(\log^2n)$ and size $O(n)$. This is easy to obtain by combining \cref{thm:Toffoli}, \cref{thm:plus_one}, and \cref{lem:rz_fanout}.
\end{proof}

\begin{figure*}[t]
    \centering
        \begin{tikzpicture}\node[scale=0.8]{
    \begin{quantikz}
        \lstick{$x_1$} &\gategroup[6, steps=14, style={mygroup}, background, label style={mylabel1}]{} &\ctrl{1} & & & &\ctrl{1} & & & & &\gate[5]{+1} & &\gate[5]{-1} &\gate{\sqrt[2^{n-1}]{Z}}\wire[d][1]{a} &\\
        \lstick{$x_2$} & &\ctrl{1} & & & &\ctrl{1} & & & & & &\gate{\sqrt[2^{n-1}]{Z}^\dag}\wire[d][1]{a} & &\gate{\sqrt[2^{n-1}]{Z}}\wire[d][1]{a} &\\
        \setwiretype{n}& &\vdots & & & &\vdots & & & & & &\vdots & &\vdots &\\
        \lstick{$x_{n-2}$} & &\ctrl{2}\wire[u][1]{a} & & & &\ctrl{2}\wire[u][1]{a} & & & & & &\gate{\sqrt[8]{Z}^\dag}\wire[u][1]{a} & &\gate{\sqrt[8]{Z}}\wire[u][1]{a} &\\
        \lstick{$x_{n-1}$} & & & &\ctrl{1} & & & &\ctrl{1} & & & &\gate{\sqrt[4]{Z}^\dag}\wire[u][1]{a} & &\gate{\sqrt[4]{Z}}\wire[u][1]{a} &\\
        \lstick{$t$} &\gate{H} &\targ{} &\gate{\sqrt[4]{Z}} &\targ{} &\gate{\sqrt[4]{Z}^\dag} &\targ{} &\gate{\sqrt[4]{Z}} &\targ{} &\gate{\sqrt[4]{Z}^\dag} &\gate{H} & & & & &\\
        \lstick{$x_n$} &\ctrl{-1} &\ctrl{-1} &\ctrl{-1} &\ctrl{-1} &\ctrl{-1} &\ctrl{-1} &\ctrl{-1} &\ctrl{-1} &\ctrl{-1} &\ctrl{-1} &\push{\heartsuit}\wire[u][2]{a} &\ctrl{-2} &\push{\heartsuit}\wire[u][2]{a} &\ctrl{-2} &
    \end{quantikz}
    };\end{tikzpicture}
    \caption{Zero ancilla construction for multi-qubit Toffoli gate. $\heartsuit$ represents for clean ancillary qubit that has value $1$. Notice that by adding to $X$ gates, it is equivalent to a clean ancillary qubit at value $0$. We omit this detail due to limitation of space.}
    \label{fig:gidney}
\end{figure*}

\section{Ancillary qubits and gate precision}\label{sec:precision}
In this section, we investigate the relationship between number of ancillary qubits and gate precision in implementations of unitary. We prove the following result: under the constraint of no ancillary qubit, both $n$-Toffoli gate and quantum incrementor cannot be implemented without using exponential precision gates. This result means that without high precision gates, our construction for $n$-Toffoli gates described in \cref{thm:Toffoli} is optimal simultaneously in circuit size, depth and number of ancillary qubits. We hope this investigation may guide future constructions of quantum circuits that is more suitable for current fault-tolerating schemes.

We say low precision gate to represent for CNOT and single-qubit gates whose phases of determinants have the form $\frac{p}{q}$ with $p,q$ intager and $q=o(2^n)$. In his construction of $n$-Toffoli gate, Gidney utilizes high precision gate $R_z(\pi/2^n)$ whose determinant is $\exp(i\pi/2^n)$ \cite{gidney2015}. It turns out to be impossible to implement $n$-Toffoli gates only using low precision gates.

\begin{theorem}
    It is impossible to implement $n$-Toffoli gate and quantum incrementor without ancillary qubits using only low precision gates. In fact, any quantum circuit implementing unitary $U$ with determinant $-1$ without ancillary qubits must have such single-qubit gate with phase of determinant $\beta\pi$ that either $\beta$ is irrational, or $\beta=\frac{p}{q}$ for some pair of coprime integers $p,q$ and $q=\Omega(2^n)$.
\end{theorem}
\begin{proof}
    Suppose the statement false and denote the circuit by $C$ which consists of only single-qubit and CNOT gates. Any CNOT gate has determinant $1$ as an $(n+1)$-qubit gate. If a single-qubit gate itself has determinant $\exp(i\frac{p}{q}\pi)$ where $p$ is coprime with $q$, then as an $(n+1)$-qubit gate, it has determinant $\exp(i\frac{2^np}{q}\pi)$; furthermore, according to assumption, the order of factor $2$ of $q$ is $o(n)$, so as an $(n+1)$-qubit gate, its determinant has phase $\frac{r}{s}\pi$ where $r$ and $s$ are coprime with $r$ even and $s$ odd.
    
    Say the circuit $C$ has $m$ single-qubit gates (as $(n+1)$-qubit gates) and denote their phases $\frac{r_1}{s_1}\pi,\dots,\frac{r_m}{s_m}\pi$. We prove this theorem by determinant argument. Since $C$ implements $U$ without ancillary qubit, $C$ has determinant $-1$. All CNOT gates have determinant $1$ and they do not contribute to the overall determinant, so the product of determinants of single-qubit gates in $C$ is equal to $-1$, which means
    \[\frac{r_1}{s_1}\pi+\dots+\frac{r_m}{s_m}\pi\equiv\pi \pmod{2\pi},\]
    that is,
    \[\frac{r_1}{s_1}+\dots+\frac{r_m}{s_m}\equiv1 \pmod{2}.\]
    Denote $a$ and $b$ the numerator and denominator of left hand side of the equation above, respectively. We know $b$ is odd because all of $s_1,\dots,s_m$ are odd. However, $a$ is even since all of $r_1,\dots,r_m$ are even. This leads to contradiction.
\end{proof}

Briefly speaking, what we do in this proof is arguing by determinant. Without ancillary qubits, both $n$-Toffoli gate and incrementor have determinant $-1$. However, adding just one ancillary qubit turns the determinant to $1$ and thus the whole unitary is in $\textup{SU}(2^{n+2})$. This also provides some intuition why our construction in \cref{sec:construction} and \cref{subsec:plus_one} works. In general, a unitary may have determinant $\exp(i\frac{p}{2^a}\pi)$ for some integer $p,a$, and adding $a$ ancillary qubits turns this unitary into a special unitary. This may be good news for exactly implementing the unitary using only low precision gates.

From the perspective described above, we conjecture that adding ancillary qubits helps reduce precision of gates used to implementing unitary operators. Yet, arguing by determinant only snoops into very limited properties of unitary. After all, we cannot even imagine adding ancillary qubit would help for implementing a simple $R_z(\pi/2^k)$ gate, while it is well-known that any single qubit gate can be approximated within $\epsilon$ error using only $O(\log\frac{1}{\epsilon})$ gates of $H$ and $T$ \cite{PhysRevLett.110.190502,ross2015optimal}. Nevertheless, we are eager to know whether there are other factors that prevent a unitary from being implemented by low precision gates.

\section{Discussion: ancillary qubits or extra energy levels}\label{sec:qudit}
In \cref{sec:construction} and \cref{subsec:plus_one}, we are able to implement $n$-Toffoli gate and quantum incrementor in $O(\log n)$ and $O(\log^2n)$ depth using $1$ ancillary qubit, respectively. But it is known that both of them can be simply implemented on qutrit system in the same depth using zero ancillary qutrit \cite{ISCAqutrit,baker2020improved}. Also in \cref{sec:precision}, we show that any circuit implementing these two gates without ancillary qubits must use exponential precision gates. However, the constructions on qutrit system consist of only reversible qutrit gates. Vaguely speaking, this may suggest that qutrit system is ``more adapted to reversibility'' than qubit system. After all, embedding a Boolean function into larger alphabets indeed stretches out more space for reversibility. But we are still not sure how to fully exploit this property.

On the other hand, there is still no evidence to the existence of such implementation task, that it is fundamentally difficult on qubit system while qudits help a lot. After all, on the level of quantum circuit, qubit system and qudit system can simulate each other efficiently. To be concrete, we remark the following simple fact.
\begin{fact}[folklore]\label{fact:exchange}
    Every quantum circuit $C$ on qutrit system can be simulated by some $C'$ on qubit system whose depth, size and qubit number are of the same order as $C$, and vice versa.
\end{fact}
\begin{proof}
    Given a quantum circuit $C$ on qutrit system on $n$ qutrits with depth $D(n)$ and size $S(n)$, we construct a quantum circuit $C'$ on qubit system on $2n$ qubits with depth $O(D(n))$ and $O(S(n))$. Now that a single qutrit has dimension $3$, each qutrit can be embedded into a two-qubit system which has dimension $4$. In $C'$, we let two adjacent qubits simulating one qutrit in $C$. Each single qutrit gate in $C$ corresponds to a two-qubit gate in $C'$, and each two-qutrit gate in $C$ corresponds to a four-qubit gate in $C'$. Finally, $C'$ can be decomposed into CNOT and single-qubit gates via any unitary synthesis technique (e.g. \cite{sun2023asymptotically}), resulting in the same order of circuit depth and size.

    Proof of the opposite direction is similar by noticing that three qubits can be simulated by two qutrits.
\end{proof}

From computation theory's perspective, all the circuits appear in this paper are $\textup{QNC}^1$ circuits where $\textup{QNC}^i$, roughly speaking, represents for unitary operations that have quantum circuits of $O(\log^in)$ depth with constant-bounded gate width \cite{moore1999quantum,hoyer2005quantum}. However, this characterization does not take the number of ancillary qubits used in implementing the unitary into account. Due to the lack of huge amount of qubits in the NISQ era, we consider it necessary to study what unitary can be implemented by low-depth circuits with bounded space.
\begin{definition}
    $\textup{QNC}^i(s)$ consists of reversible Boolean functions which can be implemented by quantum circuit that has depth $O(\log^in)$ and uses $O(s)$ clean ancillary qubits. $\textup{QNC}^i_3(s)$ is the analog of $\textup{QNC}^i(s)$ on qutrit system. Note that $\textup{QNC}^i_3(s)$ still computes Boolean functions.
\end{definition}
It is natural to ask whether $\textup{QNC}^i(s)$ hierarchy on $s$ collapses. The result in \cref{sec:precision} may suggest that there is indeed separation between using $0$ and $1$ ancillary qubits.
\begin{conjecture}
    $\textup{QNC}^i(0)\subsetneq\textup{QNC}^i(1)\subsetneq\textup{QNC}^i(n)$.
\end{conjecture}
As for the relation between ancillary qubits and extra energy levels, under this notation, what we prove in this paper is that $n$-Toffoli and incrementor belongs to $\textup{QNC}^1(1)$ and $n$-Toffoli is in $\textup{QNC}^2(0)$. In general, we have $\textup{QNC}^i(s)\subseteq\textup{QNC}^i_3(s+n)$ and $\textup{QNC}^i_3(s)\subseteq\textup{QNC}^i(s+n)$ from \cref{fact:exchange}. It is a both theoretically interesting and practical question that whether extra energy levels can really help in implementing unitary operations. In particular, we have the following conjecture:
\begin{conjecture}
    $\textup{QNC}^i_3(s)\subseteq\textup{QNC}^i(s+\textup{poly}\log n)$.
\end{conjecture}
Nevertheless, we believe that delving into this question may not only yield new techniques for designing efficient quantum circuits suitable for NISQ era, but also cast light on the nature of quantum resource assumption in quantum computation.

\bibliographystyle{alpha}
\bibliography{reference}
\end{document}